\documentclass[11pt]{article}

\usepackage[margin=1in]{geometry}

\usepackage[mathletters]{ucs}
\usepackage[utf8x]{inputenc}

\usepackage{microtype}

\usepackage{graphicx}
\usepackage[space]{grffile}

\usepackage{amsmath}
\usepackage{amsthm}

\newtheorem{theorem}{Theorem}

\newtheorem{lemma}[theorem]{Lemma}

\newtheorem*{problem*}{Problem}
\newtheorem*{theorem*}{Theorem}
\newtheorem*{lemma*}{Lemma}
\newtheorem*{proposition*}{Proposition}
\newtheorem*{definition*}{Definition}

\usepackage{amsfonts}

\usepackage[marginpar]{todo}

\usepackage{wrapfig}

\usepackage{xcolor}
\usepackage{enumerate}

\newcommand{\ignore}[1]{}

\usepackage[numbers]{natbib}

\usepackage{txfonts}

\usepackage{algorithm}
\usepackage[noend]{algpseudocode}
\usepackage{comment}

\usepackage{xspace}

\usepackage
[
colorlinks=true,
linkcolor=blue,
anchorcolor=blue,
citecolor=blue,
urlcolor=blue,
plainpages=false,
pdfpagelabels
]{hyperref}
\usepackage[capitalize]{cleveref}

\newcommand{\col}[2]{#1_{#2}}

\newcommand{\entry}[3]{\lowercase{#1}_{#2#3}}

\newcommand{\lowest}{\mathrm{pivot}}
\newcommand{\bm}{D}

\newcommand{\Z}{\mathbb{Z}}
\newcommand{\N}{\mathbb{N}}
\newcommand{\interval}[2]{\{#1,\dots,#2\}}

\DeclareMathOperator{\im}{im}

\newcommand\phat{\textsc{Phat}\xspace}
\newcommand\dipha{\textsc{Dipha}\xspace}

\title{Distributed computation of persistent homology}
\author{Ulrich Bauer\footnote{Institute of Science and Technology Austria, Klosterneuburg, Austria. \url{http://ulrich-bauer.org}}
\and Michael Kerber\footnote{Max-Planck-Center for Visual Computing and Communication, Saarbr\"ucken, Germany. \url{http://mpi-inf.mpg.de/\textasciitilde mkerber}}
\and Jan Reininghaus\footnote{Institute of Science and Technology Austria, Klosterneuburg, Austria. \url{http://ist.ac.at/\textasciitilde reininghaus}}
}

\begin{document}

\maketitle

\begin{abstract}
Persistent homology is a popular and powerful tool for capturing topological features of data.
Advances in algorithms for computing persistent homology have reduced
the computation time drastically~-- as long as the algorithm does not exhaust
the available memory. Following up on a recently presented
parallel method for persistence computation on shared memory systems,
we demonstrate that a simple adaption of the standard reduction algorithm
leads to a variant for distributed systems. Our algorithmic design
ensures that the data is distributed over the nodes without redundancy;
this permits the computation of much larger instances than on a single
machine. Moreover, we observe that the parallelism at least compensates for
the overhead caused by communication between nodes, and often even speeds up
the computation compared to sequential and even parallel shared memory algorithms.
In our experiments, we were able to compute the persistent homology of filtrations with more than a billion ($10^9$) elements within seconds on a cluster with 32 nodes using less than 10GB of memory per node.

\end{abstract}

\setcounter{page}{0}
\thispagestyle{empty}
\newpage

\section{Introduction}
\paragraph{Background}

A recent trend in data analysis is to understand the shape of data (possibly in high dimensions) using topological
methods. The idea is to interpret the data as a growing sequence of topological spaces (a \emph{filtration}), such as for example, sublevel sets of a function with an increasing threshold, or thickenings of a point set. The goal is now to compute a topological summary of the
filtration, which can be used, for instance, to identify features, as well
as to infer topological properties of a sampled shape.
Persistent homology describes how homological features
appear and disappear in the filtration (see \cref{sec:background} for more details).
Besides significant theoretical advances, persistent homology
has been used in various applications; 
see~\cite{eh-survey} for a survey.
The success of persistent homology stems from its generality, which makes
it applicable for various forms of data, 
from its stability with respect to perturbations~\cite{ccggo-proximity,ceh-stability},
from its ability to provide information on all scales,  
and, last but not least, from the availability of efficient algorithms
for computing this information.

The standard algorithm for computing persistent homology assumes the input to be a boundary matrix
of a chain complex, and proceeds by reducing that matrix using a variant of
Gaussian elimination~\cite{elz-topological,zc-computing}.
The running time is cubic in the number of simplices; this can be improved
to matrix multiplication time~\cite{mms-zigzag} or replaced by
an output-sensitive bound~\cite{ck-output}.
However, on the practical side, it has been observed that the standard
algorithm usually performs much better on real-world instances than
predicted by the worst-case bounds, and relatively simple optimizations
of the standard method yield remarkable speed-ups~\cite{ck-persistent}.
Recently, Maria et al.~\cite{mbd-compressed} implemented a memory-efficient and comparably fast method 
for computing persistent cohomology, 
which yields the same information about birth and death of features as its homology counterpart.
Further improvements have been reported by using several processors 
in a shared memory environment~\cite{bkr-clear} 
(see also~\cite{lz-multicore,lsv-spectral}
for alternative parallelization schemes). 
With these optimizations, it is often the case that computing persistence actually takes less time than
even reading the input into memory.
Therefore, the limiting factor is not so much the time spent for computation
but rather the memory available on the computer. 

\paragraph{Contribution}
We present a scalable algorithm for computing persistent homology
in parallel in a distributed memory environment. This method 
the computation of much larger instances than using existing state-of-the art algorithms
on a single machine, 
by using sufficiently many computing nodes such that the data fits into the distributed memory.
While overcoming the memory bottleneck is the primary purpose of our approach,
we aim for a time-efficient solution at the same time. 

As demonstrated by our experiments, our implementation exhibits excellent scaling with respect to memory usage on a single node, and even outperforms similar parallel shared memory code in running time.
This result is somewhat surprising, since the computation of topological properties like persistent homology is a global problem, and at first sight it is not obvious at all that the computation can be performed with the very simple and inexpensive pattern of communication that our algorithm exhibits.

Our method closely resembles 
the \emph{spectral sequence algorithm} for persistent homology~\cite[S VII.4]{eh-computational}. However, several adaptions are necessary for an efficient
implementation in distributed memory. Most importantly, reduced columns
are not stored in order of their index in the matrix, but rather
according to the order of their \emph{pivot}, the largest index of a
non-zero entry. This allows a node to perform eliminations in its associated rows, and to determine if a column with pivot in these rows is reduced,
without further communication with other nodes. Furthermore, we
minimize the number of messages sent through the network by collecting
blocks of messages, and we simplify the communication structure
by letting node $j$ only communicate with nodes $j\pm 1$.
Finally, we incorporate the clear optimization~\cite{ck-persistent}
into the algorithm in order to avoid unnecessary column operations.

\paragraph{Organization}
We introduce the necessary background on persistent homology 
in Section~\ref{sec:background}, describe our distributed memory algorithm
in Section~\ref{sec:algorithm}, report on experimental evaluation
in Section~\ref{sec:experiments}, 
and conclude in Section~\ref{sec:conclusion}.

\section{Background}
\label{sec:background}

This section summarizes the theoretical foundations of persistent homology
as needed in this work. 
We limit our scope to simplicial homology over $\Z_2$ just for the sake of simplicity in the description; our methods however easily generalize to chain complexes 
over arbitrary fields. For a more detailed introduction, we refer to \cite{eh-computational,eh-survey,zc-computing}.

\paragraph{Homology}

Homology is an algebraic invariant for analyzing the connectivity of simplical complexes.
Let $K$ be a finite simplicial complex. 
For a given dimension $d$, a \emph{$d$-chain} is a formal sum
of $d$-simplices of $K$ with $\Z_2$ coefficients. The $d$-chains
form a group $C_d(K)$ under addition.
Equivalently, a $d$-chain can be interpreted as a subset of the $d$-simplices, with the group operation being the symmetric set difference.
The \emph{boundary} of a $d$-simplex $\sigma$ is the $(d-1)$-chain formed
by the sum of all faces of $\sigma$ of codimension $1$. This operation extends linearly
to a \emph{boundary operator} $\partial_d:C_d(K)\rightarrow C_{d-1}(K)$.
A $d$-chain $\gamma$ is a \emph{$d$-cycle} if $\partial_d(\gamma)=0$.
The $d$-cycles form a subgroup of the $d$-chains, denoted by $Z_d(K)$. A $d$-chain $\gamma$ is called a 
\emph{$d$-boundary} if $\gamma=\partial_d(\xi)$ for some $(d+1)$-chain $\xi$.
Again, the $d$-boundaries form a subgroup $B_d(K)$ of $C_d(K)$, and since $\partial_d(\partial_d(\xi))=0$
for any chain $\xi$, $d$-boundaries are $d$-cycles, 
and so $B_d(K)$ is a subgroup of $Z_d(K)$.
The \emph{$d^\text{th}$ homology group} $H_d(K)$ is defined as the quotient group
$Z_d(K)/B_d(K)$. 
In fact, the groups $C_d(K)$, $Z_d(K)$, $B_d(K)$, and $H_d(K)$ are $\Z_2$-vector spaces. The dimension of $H_d(K)$ is called
the \emph{$d^\text{th}$ Betti number} $\beta_d$.
Roughly speaking, the Betti numbers in dimension $0$, $1$, and $2$ count the number
of connected components, tunnels, and voids of $K$, respectively.

\paragraph{Persistence}
Consider a \emph{simplexwise filtration} of $K$, i.e., a sequence of inclusions
$\emptyset=K_0\subset\ldots\subset K_n=K$ 
such that $K_i=K_{i-1}\cup\{\sigma_i\}$, where $\sigma_i$ is a simplex of $K$.
We write $H_*(K_i)$
for the direct sum of the homology groups of $K_i$ in all dimensions.
For $i\leq j$, the inclusion $K_{i}\hookrightarrow K_{j}$ induces a homomorphism
$h_{i}^{j}: H_*(K_i)\rightarrow H_*(K_{j})$ on the homology groups. 
We say that a class $\alpha\in H_*(K_i)$ is
\emph{born at (index) $i$} if \[\alpha\notin\im h_{i-1}^i. \]
A class $\alpha\in H_*(K_i)$ born at index $i$ \emph{dies entering (index)~$j$} if
\[h_i^{j-1}(\alpha)\notin\im h_{i-1}^{j-1}\text{ but }h_i^j(\alpha)\in\im h_{i-1}^j
.\]
In this case, the index pair $(i,j)$ is called a \emph{persistence pair}, and the difference $j-i$ is the \emph{(index) persistence}
of the pair. 
The transition from $K_{i-1}$ to $K_i$ either causes the birth or the death
of some homology class. This homology class is not unique in general.

\paragraph{Boundary matrix}
For a matrix $M\in(\Z_2)^{n\times n}$, let $\col{M}{j}$
denote its $j^\text{th}$ column and $\entry{M}{i}{j}\in\Z_2$
its entry in row $i$ and column $j$.
For a column $\col{M}{j}$, %
we define $\lowest(\col{M}{j})=\min\{p\in \N_0: \entry{M}{i}{j}=0\text{ for all }i>p\}$
and call it the \emph{pivot index} of that column. When obvious from the context, we omit explicit mention of the matrix $M$ and write $\lowest(j)$ for $\lowest(\col{M}{j})$.

The \emph{boundary matrix} $\bm\in(\Z_2)^{n\times n}$ 
of a simplexwise filtration 
$(K_i)_{i=1}^n$ is the $n\times n$ matrix of the boundary operator $\partial_*:C_*(K)\to C_*(K)$ with respect to the ordered basis $(\sigma_i)_{i=1}^n$ of $C_*(K)$. We have
$\entry{\bm}{i}{j}=1$ 
if and only if $\sigma_i$ is a face of $\sigma_j$ of codimension~$1$.
In other words, the $j^\text{th}$ column of $\bm$ encodes the boundary
of $\sigma_j$. 
$\bm$ is an upper-triangular matrix, since any face of $\sigma_j$
must precede $\sigma_j$ in the filtration.

\paragraph{Matrix reduction}
A column operation of the form
$\col{M}{j}\gets \col{M}{j}+\col{M}{k}$ is called \emph{left-to-right addition} if $k<j$. 
A left-to-right addition is called \emph{eliminating} if it decreases $\lowest(j)$.
A column $\col{M}{j}$ is called
\emph{reduced} if $\lowest(j)$ cannot be decreased by applying any sequence of left-to-right additions.
In particular, there is no non-zero column $\col{M}{k}$ with $k<j$ and $\lowest(k)=\lowest(j)$. Clearly a zero column is reduced.
Note that a reduced column remains reduced under eliminating left-to-right column additions.

We call a matrix $M$ \emph{reduced} if all columns are reduced, or equivalently, if no two non-zero columns have the same 
pivot index. We call $M$ \emph{reduced at index $(i,j)$} 
if the lower left
submatrix of $M$ with rows of index~$> i$ and columns of index $\leq j$ is reduced.
A sufficient condition for column $\col{M}{j}$ to be reduced is that $M$ is reduced at index $(i,j)$ with $i=\lowest(j)$.

If $R$ is a reduced matrix obtained by applying left-to-right additions to $M$, we call it a \emph{reduction of $M$}.
In this case, we define
\[
P_R:=\{(i,j)\mid i=\lowest(\col{R}{j})>0\}\\
\]
Although the reduction matrix $R$ is not unique, 
the set $P_R$ is the same for any reduction of $M$;
therefore, we can define
$P_{M}$ to be equal to $P_R$  for any reduction $R$
of $M$. %

\paragraph{Persistence by reduction}

For the boundary matrix $\bm$ of the filtration $(K_i)_{i=1}^n$, 
the first $i$ columns generate the boundary group $B_*(K_i)$. This property is invariant under left-to-right column additions.
For a reduction $R$ of $\bm$, 
the non-zero columns among the first $i$ columns actually form a basis of $B_*(K_i)$. 
Note that 
\begin{align*}i=\dim C_*(K_i)&=\dim Z_*(K_i)+\dim B_*(K_i)\text{ and }\\
\dim H_*(K_i)&=\dim Z_*(K_i)-\dim B_*(K_i).
\end{align*}
Hence, if $\col{R}{i}$ is zero, we have 
\begin{align*}
\dim B_*(K_i)&=\dim B_*(K_{i-1}), \\
\dim Z_*(K_i)&=\dim Z_*(K_{i-1})+1 \text{, and}\\
\dim H_*(K_i)&=\dim H_*(K_{i-1})+1,
\end{align*}
and so some homology class is born at $i$. If on the other hand $\col{R}{j}$ is non-zero with $i=\lowest(j)$, we have
\begin{align*}
\dim B_*(K_i)&=\dim B_*(K_{i-1})+1, \\
\dim Z_*(K_i)&=\dim Z_*(K_{i-1}) \text{, and}\\
\dim H_*(K_i)&=\dim H_*(K_{i-1})-1.
\end{align*}
The fact that $\col{R}{j}$ has pivot $i$ means that $\col{R}{j}\in Z_*(K_i)$ and hence \[[\col{R}{j}]\in H_*(K_i);\] the fact that it is reduced means that there is no $b\in B_*(K_{j-1})$ with $b+\col{R}{j}\in Z_*(K_{i-1})$ and hence \[[\col{R}{j}]\not\in\im h_{i-1}^i.\] We conclude that $[\col{R}{j}]$ is born at $i$. We even have \[[\col{R}{j}]\not\in\im h_{i-1}^{j-1}.\] Moreover, $\col{R}{j}$ is a boundary in $K_{j-1}$, and so \[[\col{R}{j}]=0\in\im h_{i-1}^{j}.\] We conclude that
the pairs $(i,j)\in P_\bm$ are the persistence pairs
of the filtration.

The standard way to reduce $\bm$ is to process columns from left 
to right; for every column, previously reduced columns are added from the left
until the pivot index is unique.
A lookup table can be used
to identify the next column to be added in constant time. 
The running time is at most cubic in $n$,
and this bound is actually tight for certain input filtrations,
as demonstrated in~\cite{morozov-persistence}. 

\paragraph{Clearing optimization}

Despite its worst-case behavior, there are techniques to speed up
the reduction significantly in practice. A particularly simple yet powerful improvement
has been presented in~\cite{ck-persistent}. It is based on the following observations.

First, the reduction of the matrix can be performed separately for each dimension $d$, by restricting to the submatrix corresponding to columns of dimension $d$ and rows of dimension $d-1$. This submatrix is exactly the matrix of the boundary operator $\partial_d:C_p(K)\rightarrow C_{p-1}(K)$.
The second basic fact to note is that in any reduction of $\bm$,
if $i$ is a pivot of some column $j$, the $i^\text{th}$ column is zero.

This leads to the following variant of the reduction algorithm:
the boundary matrix is reduced separately in each dimension in decreasing order.
After the reduction in dimension $d$, 
all columns corresponding to pivots indices are set to zero~-- we call 
this process \emph{clearing}. Note that columns 
corresponding to $d$-simplices have pivots corresponding to $d-1$-simplices.
After clearing, we proceed with the reduction in dimension $d-1$. 

\section{Algorithm}
\label{sec:algorithm}

Throughout the section, 
let $(K_i)_{i=1}^n$ be a filtration of a simplicial complex 
consisting of $n$ simplices, represented by its boundary matrix $\bm$.
Our goal is to compute the persistence pairs of $(K_i)_i$ on a cluster of $p$
processor units, called \emph{nodes}, 
which are indexed by the integers $\interval{1}{p}$.

\paragraph{Reduction in blocks} 
Let $0=r_0<\dots<r_i<\dots<r_p=n$ be an integer partition of the interval $\interval{0}{n}$.
Let the \emph{$i^\text{th}$ range} be the interval 
of integers $k$ with $r_{i-1}<k\leq r_i$.
We define the \emph{block $(i,j)$} of~$M$
as the block submatrix
with rows from the $i^\text{th}$ row range and columns from the $j^\text{th}$ columns range.
The blocks partition the matrix into $p^2$ submatrices. Any block $(i,j)$ with $i>j$ is completely zero, since $\bm$ is lower triangular.

To simplify notation, we call $M$ \emph{reduced at block $(i,j)$} if $M$ is reduced at index $(r_{i-1},r_j)$. Moreover, we call $M$ \emph{reducible in block $(i,j)$} 
if $M$ is reduced at block $(i,j-1)$ and at block $(i+1,j)$. 
This terminology is motivated by the fact that in order to obtain a matrix that is reduced at block $(i,j)$, only entries in block $(i,j)$ have to be eliminated, as described in \cref{alg:block_reduction} and shown in the following lemma. 

\begin{algorithm}[ht]
\caption{Block reduction}
\label{alg:block_reduction}
\begin{algorithmic}[1]

\Require{input $M$ is reducible in block $(i,j)$}
\Ensure{result $M$ is reduced at block $(i,j)$}
\Procedure {ReduceBlock}{$i,j$}
    \For{each $l$ in range $j$ in increasing order}
    \label{alg:main_for_loop}
        \While{$\exists k$ with $\lowest(k)=\lowest(l)$ in range $i$}
        \label{alg:main_while_loop}
          \State add column $k$ to column $l$
        \EndWhile
        \If{$\lowest(l)$ is in range $i$}
        	\State add column $l$ to collection of reduced columns
        \EndIf
    \EndFor
\EndProcedure
\end{algorithmic}
\end{algorithm}
\begin{lemma}\label{lem:block_reduction}
\Cref{alg:block_reduction} is correct: if $M$ is reducible in block $(i,j)$, then applying \cref{alg:block_reduction} yields a matrix which is reduced at block $(i,j)$.
\end{lemma}
\begin{proof}
By induction on $l$, $M$ is reduced at index $(r_{i-1},l)$ after each iteration of the main \textbf{for} loop (\cref{alg:main_for_loop}). This follows directly from the exit condition of the \textbf{while} loop in \cref{alg:main_while_loop}, together with the induction hypothesis and the precondition that $M$ is reduced at index $(r_i,r_j)$ and hence also at index $(r_i,l)$.
\end{proof}
\begin{lemma}
\Cref{alg:block_reduction} only requires access to the unreduced columns of $M$ in range $j$ and the reduced columns with pivot in range $i$.
\end{lemma}
\begin{proof}
Let $l$ be in range $j$ and let $k<l$ be such that $\lowest(k)=\lowest(l)$ is in range $i$, as in the \textbf{while} loop in \cref{alg:main_while_loop}. Then clearly column $l$ is unreduced. Moreover, as shown in the proof of \cref{lem:block_reduction}, $M$ is reduced at index $(r_{i-1},l)$. Since $\lowest(k)$ is in range $i$, we have $r_{i-1}<\lowest(k)$, and by assumption $k<l$. Hence $M$ is also reduced at index $(\lowest(k),k)$, i.e., column $k$ is reduced.
\end{proof}

\paragraph{Parallel reduction}
We now describe a parallel algorithm to reduce a boundary matrix $D$ by applying block reduction on all blocks $(i,j)$ with $i\leq j$ in a certain order.

The algorithm 
reduces the blocks starting with the diagonal blocks $(i,i)$
with $1\leq i\leq p$. Indeed, note that the boundary matrix $\bm$ is $(i,i)$-reducible
for any diagonal block $(i,i)$. All block reductions for diagonal blocks are independent
and can be performed in parallel. Now consider a block of the form $(i,j)$ with $i<j$.
Note that this block can be reduced as soon as blocks the $(i,j-1)$ and $(i-1,j)$ have been reduced. This relation defines a partial order on the blocks $(i,j)$ with $i\leq j$. If the order of execution of the block reductions is consistent with that partial order, the preconditions of block reduction are satisfied in every block. 
Note that two blocks $(i,j)$ and $(i',j')$ can be reduced independently iff either ($i<i'$ and $j<j'$) or ($i>i'$ and $j>j'$). After having reduced the block $(1,p)$, the postcondition of \cref{alg:block_reduction} yields that the resulting matrix is a reduction of the input boundary matrix $D$.

Note that a special case of this block-reduction scheme is the \emph{spectral sequence algorithm} presented in~\cite[S VII.4]{eh-computational}. This algorithm sweeps the blocks diagonally, and in each phase $r\in\interval1p$ of the sweep it reduces all blocks $(i,j)$ with $j-i=r-1$ in order of increasing index $i$. The algorithm as described is sequential, however, as discussed above, within a given phase $r$ the blocks can be reduced independently.

\paragraph{Distributed reduction}
We now describe how the data and the workload are distributed and transferred between the nodes.

Each node $i$ is assigned a row of blocks for reduction. The blocks are necessarily processed from left to right.
Recall that reducing a block $(i,j)$ requires access to the unreduced columns in range $j$, and to the reduced columns with pivot in range $i$. 
During the execution, each node $i$ maintains a collection of all reduced columns with pivot in the $i^\text{th}$ range, indexed by pivot. The unreduced columns in a given range $j$, on the other hand, are passed on from node to node. No data is duplicated among the nodes; each column of the matrix is stored in exactly one node throughout the execution of the algorithm. The union of the locally stored unreduced and reduced columns yields a distributed representation of the partially reduced
boundary matrix.

Initially, each node $i$ loads the columns of the input boundary matrix in range $i$. 
The following procedure is now repeated, with $j$ ranging from $i$ to $m$.
Node $i$ performs reduction in block $(i,j)$ and retains the reduced columns with pivot in range $i$ in its collection.
After that, it sends a package to node $i-1$ containing the remaining unreduced columns in range $j$ (if $i>1$), and receives a package from node $i+1$ containing the unreduced columns in range $j+1$ (if $j<p$).

\begin{algorithm}[ht]
\caption{Distributed matrix reduction}
\label{alg:distributed_reduction}
\begin{algorithmic}[1]

\Require{access to columns of input boundary matrix $D$ in range $j$}
\Ensure{resulting output matrix $R$ is a reduction of $D$}
\Procedure {ReduceOnNode}{$i$}
	\State input package with columns of $D$ in range $i$
	\For{$j=i,\ldots,p$}
		\State \Call {ReduceBlock}{$i,j$}
		\If {$i>1$}
			\State send package with unreduced columns in range $j$ to node $i-1$
		\EndIf
		\If {$j<p$}
			\State receive package with unreduced columns in range $j+1$ from node $i+1$
		\EndIf
	\EndFor
	\State \Return reduced columns with pivot in range $i$
\EndProcedure
\end{algorithmic}
\end{algorithm}

Observe that in each iteration, node $i$ has all the information required to perform reduction in block $(i,j)$, namely, the unreduced columns in range $j$ and the reduced columns with pivot in range $i$. Moreover, the preconditions for block reduction are satisfied, since block $(i,j-1)$ is reduced on the same node $i$ before block $(i,j)$, and block $(i+1,j)$ is reduced on node $i+1$ before node $i$ receives the unreduced columns in range $j$ from node $i+1$. We conclude:
\begin{lemma}\label{lem:distributed_reduction}
If \cref{alg:distributed_reduction} is executed on a cluster with $p$ nodes, it computes a reduction of the input matrix.
\end{lemma}

Note that the structure of communication between the nodes is very simple: each node $i$ only receives data from node $i+1$ and only sends data to node $i-1$. Moreover, less than $p$ messages are sent between each pair of consecutive nodes. This is highly beneficial for distributed computing, as the communication overhead and the network latency become negligible.

\paragraph{Clearing in parallel}
The clearing optimization from Section~\ref{sec:background} can be implemented
in the distributed reduction algorithm with minor changes. 
Recall that the the clearing optimization iterates over the dimensions $d$ in decreasing order and processes only the columns of a given dimension $d$ at a time.

The ranges are
defined by a single global partition $r_0<\ldots<r_p$ 
that does not change per dimension.
Note that this might cause initial column packages of different sizes 
in a given dimension, even if the ranges are all of same size.
However, it has the following advantage: when node $i$ has performed
its last block reduction for dimension $d$, it knows all 
pivots that fall in the $i^\text{th}$ range. All these pivots corresponds 
to $d-1$-simplices that create homology and hence correspond to zero columns in any reduction.
In the next iteration, node $i$ is initialized to process the columns
of dimension $d-1$ in the $i^\text{th}$ range. Before it starts
the block reduction, it can simply clear all columns with indices
that were pivots in dimension $d$. In particular, no communication with other
nodes is required.

\paragraph{Design rationale} 
We justify some design choices in our algorithm and discuss alternatives.
First, we implemented the sending of packages 
in Algorithm~\ref{alg:distributed_reduction} in a \emph{blocking} fashion,
i.e., a node does not start receiving the next package
until has sent and discarded its current package. Clearly, this strategy can result
in delayed processing of packages because a sending node has to wait
for its predecessor to be ready to receive a package. On the other hand, the strategy guarantees that
every node holds at most one package at a time; this prevents a slower node
from accumulating more and more packages, possibly causing high memory consumption.

A possible strategy to reduce the overall amount of communication would be
to have node $i$ send a unreduced column with pivot
in the $k^\text{th}$ range to node $k$ directly, instead of the predecessor node $i-1$.
However, this approach would complicate
the communication structure and data management significantly. Any node would have to be
able to receive unreduced columns any time, and it would not be possible to bound the number of unprocessed columns a node has to maintain in memory. It would also increase the number of messages send through
the network.

A somewhat dual approach to our communication scheme would be to send the reduced
columns from node $i$ to $i+1$ instead of sending the unreduced columns from node $i$ to $i-1$.
In this variant, node $i$ would perform reduction 
in block $(j,i)$ for $j=i,i-1,\ldots,1$.
However, in this approach, the package size would increase towards
the end of the reduction, as the number of reduced columns increases, whereas in
our implementation the package size decreases together with the number of reduced columns.
Since typically most columns are reduced early on, we expect much more data to be sent 
between the nodes using this variant.

\section{Experiments}
\label{sec:experiments}

\begin{table}
\centering
\begin{tabular}{r||r|r||r|r|r|r|r|r|r|r}
&\multicolumn{2}{c||}{\phat}&\multicolumn{5}{c|}{\dipha}\\
\hline
cores/nodes&1&16&2&4&8&16&32\\
\hline
GRF2-$256$&10.2GB&10.5GB&11.1GB&5.6GB&2.8GB&1.4GB&0.74GB\\
GRF1-$256$&10.8GB&11.3GB&11.8GB&6.1GB&3.1GB&1.5GB&0.8GB\\
GRF2-$512$&&&&&&11.1GB&5.7GB\\
GRF1-$512$&&&&&&&9.1GB\\
vertebra16&&&&&&&9.0GB
\end{tabular}
\caption{Peak memory consumption for sequential and parallel shared memory
(\phat, left) algorithms and our distributed algorithm (\dipha, right)}
\label{tab:peak_mem}
\begin{tabular}{r||r|r||r|r|r|r|r|r|r|r}
&\multicolumn{2}{c||}{\phat}&\multicolumn{5}{c|}{\dipha}\\
\hline
cores/nodes&1&16&2&4&8&16&32\\
\hline
GRF2-$256$&14.6s&5.2s&10.1s&5.5s&3.4s&2.2s&1.6s\\
GRF1-$256$&28.8s&12.8s&27.2&20.3&15.4&12.1s&9.9s\\
GRF2-$512$&&&&&&17.9s&11.2s\\
GRF1-$512$&&&&&&&95.3s\\
vertebra16&&&&&&&34.9s
\end{tabular}
\caption{Running times for sequential and parallel shared memory
(\phat, left) algorithms and our distributed algorithm (\dipha, right)}
\label{tab:time}
\end{table}

Since our algorithm is, to the best of our knowledge, the first attempt at computing
persistence in a distributed memory context, we concentrate our experimental
evaluation on two aspects. 
First, how does our approach scale with an increasing
number of nodes, in running time and memory consumption? 
Second, how does the our algorithm compare with state-of-the-art 
sequential and parallel shared memory implementations on instances which are 
still computable in this context?

We implemented Algorithm~\ref{alg:distributed_reduction} in C++
using the OpenMPI implementation of the Message Parsing Interface standard%
\footnote{\url{www.open-mpi.org}}.
We ran the distributed algorithm on a cluster with up to 32 nodes, each with two Intel Xeon CPU E5-2670 2.60GHz processors (8 cores each) and 64GB RAM, connected by a 40Gbit Infiniband interconnect.

For comparison, our tests also include results for the \phat library%
\footnote{\url{http://phat.googlecode.com}},
which contains efficient sequential and parallel shared memory algorithms
for computing persistence. Among the sequential versions, 
the \verb"--twist" algorithm option, which
is the standard reduction with the clearing optimization described
in Section~\ref{sec:background}, together with the \verb"--bit_tree_pivot_column"
data structure option, showed the overall best performance (see the \phat documentation for more information). For parallel shared memory, the \verb"--block_spectral_sequence" algorithm with the \verb"--bit_tree_pivot_column"
data structure showed the overall best performance on the tested examples.
We therefore used these two variants for comparison.
The sequential and parallel shared memory algorithms were run on a single machine of the cluster.
In order to obtain a clear comparison between the shared memory and distributed memory algorithms, 
in our test of the distributed algorithm only one processor core per node was used.

For our tests, we focus on filtrations induced by 3D image data.
In particular, we used isotropic Gaussian random fields whose power spectral density is given by a power law $\|x\|^{-p}$. This process is commonly used in physical cosmology as a model for cosmic microwave background~\cite{p-cosmo}. We consider two images sizes: filtrations of images of size $256^3$ have a length of $n=511^3\approx133$ millions and a binary file size of around
5GB, while images of size $512^3$ yield a filtration of length $n=1023^3\approx1.07$ billions and a file size of around 40GB.
In addition, we included the $512^3$ medical image \verb"vertebra16" from the \textsc{VolVis} repository\footnote{Available at \url{http://volvis.org}} in our test set, a rotational angiography scan of a head with an aneurysm.

\begin{figure}[t]
\centering
\includegraphics[width=0.3\textwidth]{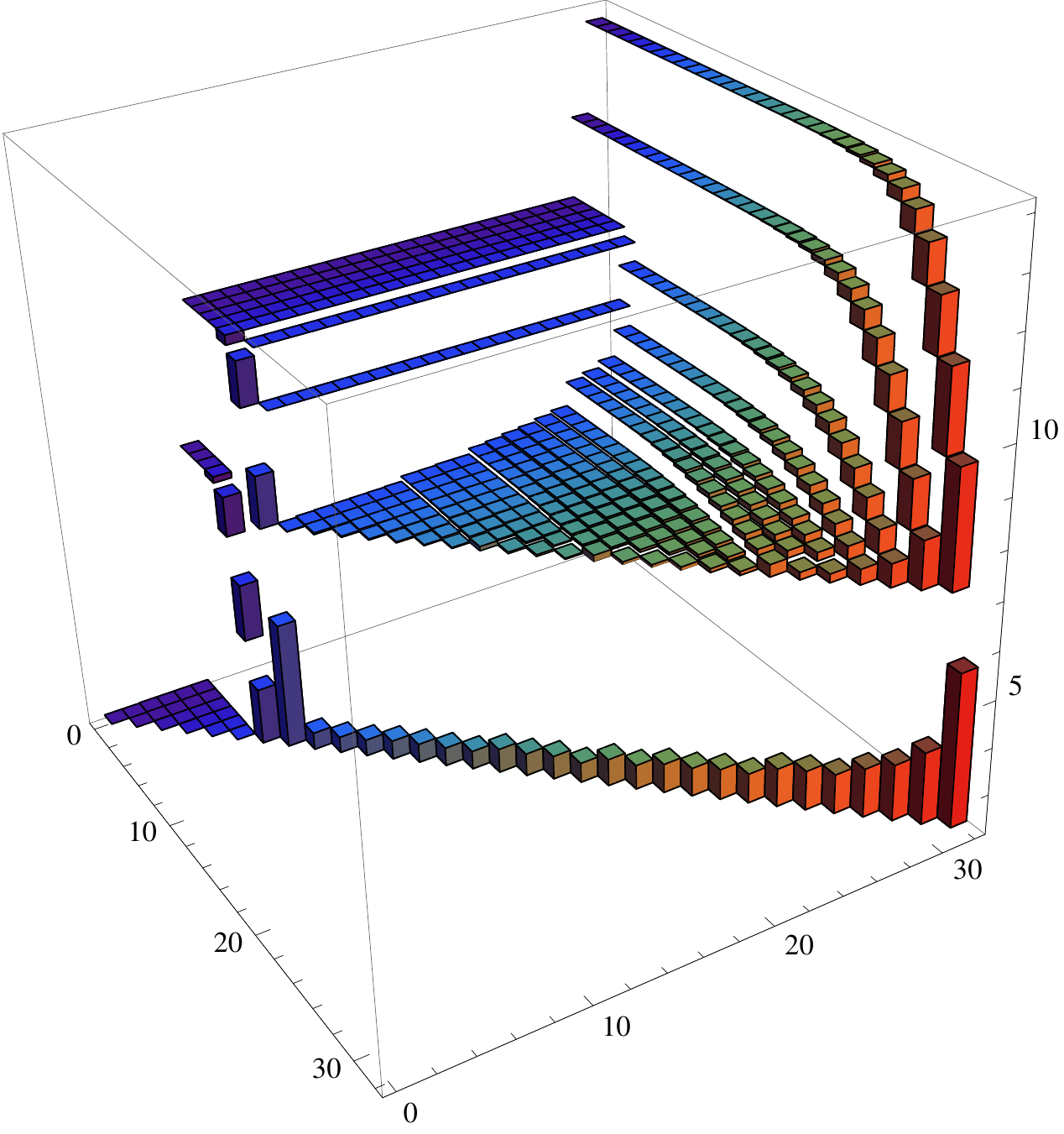}
\includegraphics[width=0.3\textwidth]{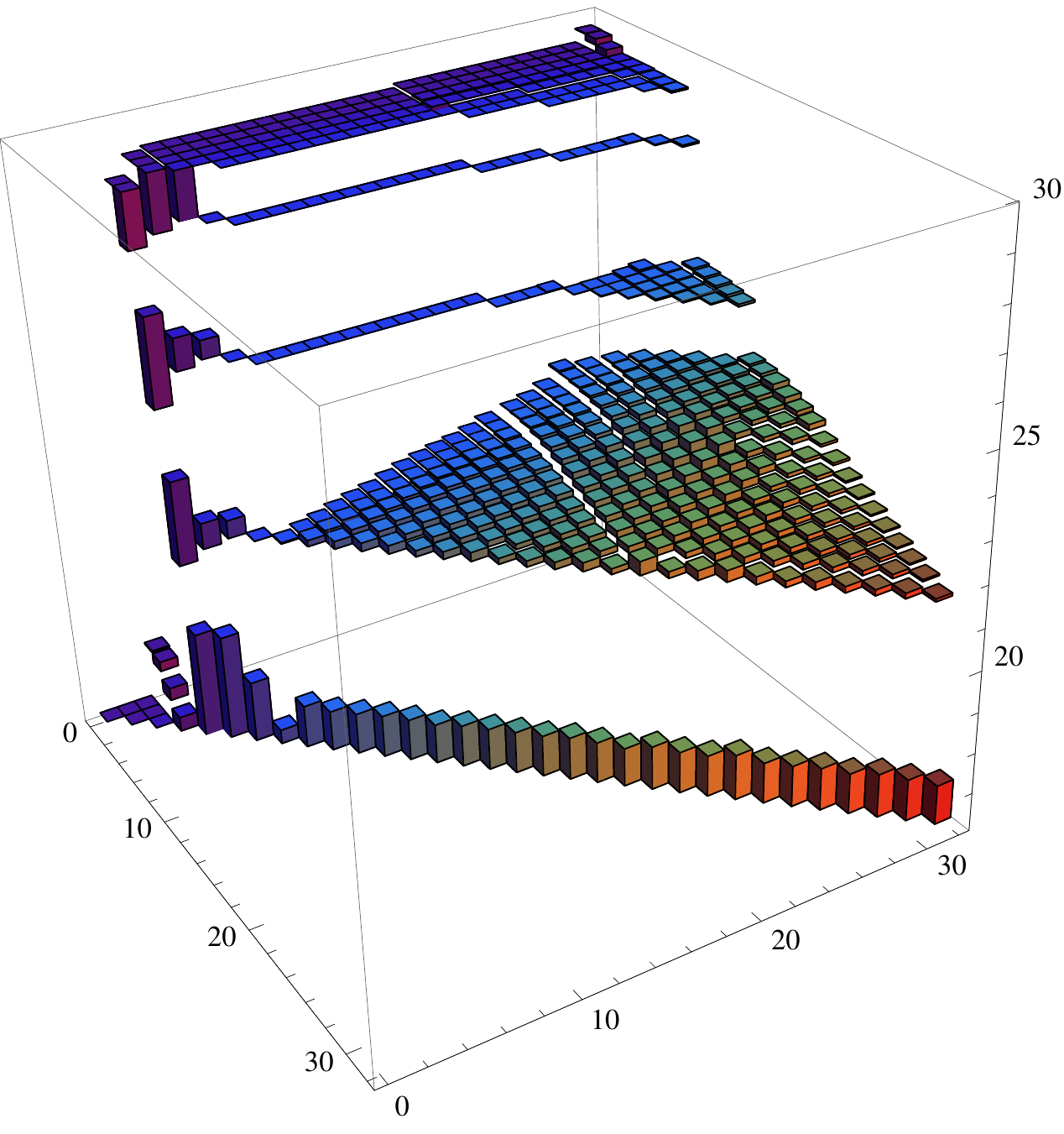}
\includegraphics[width=0.3\textwidth]{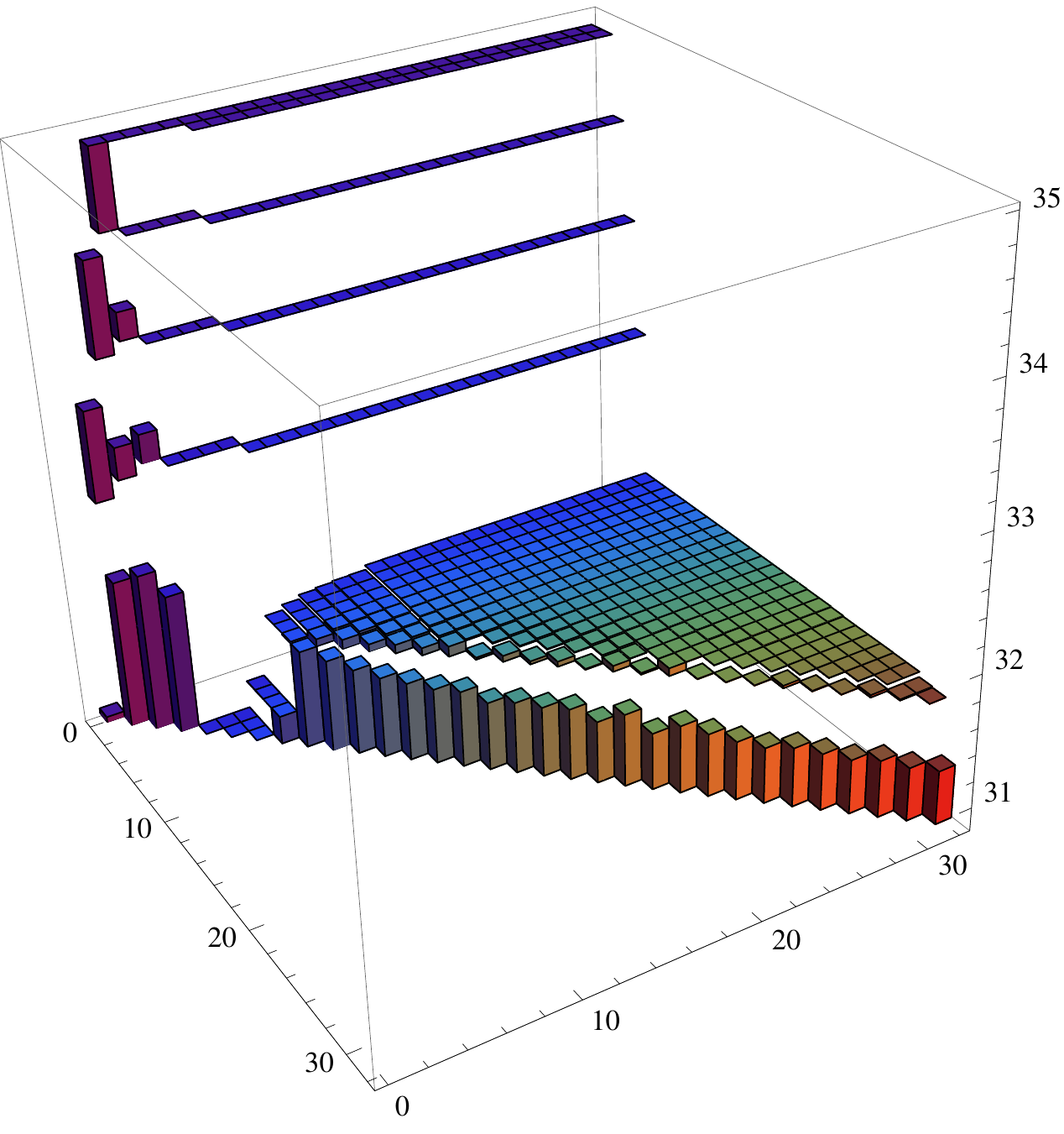}
\caption{Running times for each block reduction in dimensions $\delta=3,2,1$ for the vertebra16 data set.}
\label{tab:running_time}
\end{figure}

\paragraph{Scalability}
\Cref{tab:peak_mem,tab:time} show the running time and peak memory consumption
of our algorithm for images of size $256^3$ and $512^3$.
The table is incomplete; the algorithm was not able to compute a result
for the remaining cases because of address space limitations of the OpenMPI I/O API that we plan 
to circumvent in a forthcoming version.
We observe that the memory usage per node
is almost exactly halved when doubling the number of nodes. For
the running time, the speed-up factor is not quite as high,
but still the algorithm terminates faster when using more
nodes. In summary, this provides strong evidence that our algorithm 
scales well with the number of nodes, both regarding time and space complexity.

\paragraph{Comparison}
\Cref{tab:peak_mem,tab:time} also lists the results for the best sequential
and parallel shared memory algorithms of the \phat library. Both algorithms run out of 
memory when trying to compute persistence for larger examples on our testing 
machine, showing that our distributed approach indeed extends
the set of feasible instances. Moreover, we observe that the 
running time on 16 nodes with distributed memory is actually
lower than that of the parallel shared memory algorithm on a single machine
with 16 processor cores. One reason might be that the distributed system 
has a much larger total amount of processor cache available than the 
shared memory system. Since matrix reduction is more memory intensive than
processor intensive, this effect may actually outweigh the overhead of 
communication over the network.
This suggests that the distributed approach may be preferable
even if the solution is in principle computable in a non-distributed
environment.

\paragraph{Communication analysis}
We give more details on the amount of data transmitted between the nodes by our algorithm.
Table~\ref{tab:total_communication} shows the total amount of data
 exchanged; \cref{tab:pair_communication} shows the largest total amount of data
 transmitted between any pair of nodes; \cref{tab:package_communication} shows the largest package size. For the more challenging examples, the amount
is in the range of GBs. Considering the bandwidth of modern interconnects and the fact that communication is bundled in a small number of packages, the running time of the local block reductions dominates the time spent for communication.
This is illustrated in \cref{tab:running_time}, which shows a plot of the running times for each block reduction for the vertebra16 data set on 32 nodes.

\begin{table}[t]
\centering
\begin{tabular}{r||r|r|r|r|r|r|r|r}
nodes&2&4&8&16&32\\
\hline
GRF2-$256$&5.6MB&15.1MB&32.5MB&67.7MB&136MB\\
GRF1-$256$&69.2MB&218MB&497MB&1.0GB&2.0GB\\
GRF2-$512$&&&&342MB&694MB\\
GRF1-$512$&&&&&34.0GB\\
vertebra16&&&&&19.1GB\\
\end{tabular}
\caption{Total size of all packages sent over the network}
\label{tab:total_communication}

\bigskip

\begin{tabular}{r||r|r|r|r|r|r|r|r}
nodes&2&4&8&16&32\\
\hline
GRF2-$256$&5.6MB&5.6MB&5.6MB&6.5MB&8.7MB\\
GRF1-$256$&69.2MB&90.3MB&109MB&162MB&238MB\\
GRF2-$512$&&&&29.2MB&29.7MB\\
GRF1-$512$&&&&&5.0GB\\
vertebra16&&&&&4.2GB
\end{tabular}
\caption{Maximum total size of all packages transmitted between any pair of nodes}
\label{tab:pair_communication}

\bigskip

\begin{tabular}{r||r|r|r|r|r|r|r|r}
nodes&2&4&8&16&32\\
\hline
GRF2-$256$&3.1MB&2.9MB&2.7MB&3.6MB&2.5MB\\
GRF1-$256$&61.4MB&52.0MB&69.0MB&50.9MB&38.4MB\\
GRF2-$512$&&&&11.5MB&9.0MB\\
GRF1-$512$&&&&&1.9GB\\
vertebra16&&&&&1.5GB\\
\end{tabular}
\caption{Maximum package size sent over the network}
\label{tab:package_communication}
\end{table}

\section{Conclusion}
\label{sec:conclusion}
We presented the first implementation of an algorithm for computing persistent homology
in a distributed memory environment. While our algorithm resembles
the spectral sequence algorithm for persistence computation to a
large extent, 
several lower-level design choices were necessary for an efficient realization.
Our approach permits the computation of instances that were infeasible
for previous methods, and the parallelism also speeds up the
computation for previously feasible instances.

We plan to extend our experimental evaluation in future work. One problem
in benchmarking our new approach is that persistence
computation is only the second step in the pipeline: first, one has
to generate a filtration that serves as the input for the algorithm.
This itself usually requires a massive computation, which at some point becomes infeasible on single machines as well. 
We are currently working on methods for generating filtrations of large 3D images and Rips filtrations
in a distributed memory environment.
\newpage

\bibliography{bib}

\begin{thebibliography}{10}

\bibitem{bkr-clear}
Ulrich Bauer, Michael Kerber, and Jan Reininghaus.
\newblock Clear and compress: Computing persistent homology in chunks.
\newblock In {\em TopoInVis 2013}, 2013.

\bibitem{ccggo-proximity}
F.~Chazal, D.~Cohen-Steiner, M.~Glisse, L.~Guibas, and S.~Oudot.
\newblock Proximity of persistence modules and their diagrams.
\newblock In {\em Proc.\ 25th ACM Symp.\ on Comp.\ Geom.}, pages 237--246,
  2009.

\bibitem{ck-output}
Chao Chen and Michael Kerber.
\newblock An output-sensitive algorithm for persistent homology.
\newblock In {\em Proceedings of the 27th Annual Symposium on Computational
  Geometry}, pages 207--215, 2011.

\bibitem{ck-persistent}
Chao Chen and Michael Kerber.
\newblock Persistent homology computation with a twist.
\newblock In {\em 27th European Workshop on Computational Geometry (EuroCG)},
  pages 197--200, 2011.
\newblock Extended abstract.

\bibitem{ceh-stability}
D.~Cohen-Steiner, H.~Edelsbrunner, and J.~Harer.
\newblock Stability of persistence diagrams.
\newblock {\em Discrete and Computational Geometry}, 37:103--120, 2007.

\bibitem{eh-computational}
H.~Edelsbrunner and J.~Harer.
\newblock {\em Computational Topology, An Introduction}.
\newblock American Mathematical Society, 2010.

\bibitem{eh-survey}
Herbert Edelsbrunner and John Harer.
\newblock Persistent homology — a survey.
\newblock In Jacob~E. Goodman, J\'{a}nos Pach, and Richard Pollack, editors,
  {\em Surveys on Discrete and Computational Geometry: Twenty Years Later},
  Contemporary Mathematics, pages 257--282. 2008.

\bibitem{elz-topological}
Herbert Edelsbrunner, David Letscher, and Afra Zomorodian.
\newblock Topological persistence and simplification.
\newblock {\em Discrete and Computational Geometry}, 28:511--533, 2002.

\bibitem{lz-multicore}
R.~H. Lewis and A.~Zomorodian.
\newblock Multicore homology.
\newblock Manuscript, 2012.

\bibitem{lsv-spectral}
David Lipsky, Primoz Skraba, and Mikael Vejdemo-Johansson.
\newblock A spectral sequence for parallelized persistence.
\newblock arXiv:1112.1245, 2011.

\bibitem{mbd-compressed}
Clement Maria, Jean-Daniel Boissonnat, and Tamal Dey.
\newblock The compressed annotation matrix: An efficient data structure for
  computing persistent cohomology.
\newblock In {\em ESA 2013}, 2013.

\bibitem{mms-zigzag}
Nikola Milosavljevi\'{c}, Dmitriy Morozov, and Primo\v{z} \v{S}kraba.
\newblock Zigzag persistent homology in matrix multiplication time.
\newblock In {\em Proceedings of the 27th Annual Symposium on Computational
  Geometry}, pages 216--225, 2011.

\bibitem{morozov-persistence}
Dmitriy Morozov.
\newblock Persistence algorithm takes cubic time in the worst case.
\newblock In {\em BioGeometry News}. Duke Computer Science, Durham, NC, 2005.

\bibitem{p-cosmo}
John Peacock.
\newblock {\em Cosmological Physics}.
\newblock Cambridge University Press, 1999.

\bibitem{zc-computing}
Afra Zomorodian and Gunnar Carlsson.
\newblock Computing persistent homology.
\newblock {\em Discrete and Computational Geometry}, 33:249--274, 2005.

\end{thebibliography}
\bibliographystyle{plain}

\end{document}